\documentclass[11pt]{article}   \usepackage{ASA_preamble}

\title{Formal exponentials and linearisations of QP-manifolds}

\author{Alex S.~Arvanitakis\thanks{Email: \texttt{alex.s.arvanitakis@vub.be}} }
\affil{\small
\textit{Theoretische Natuurkunde, Vrije Universiteit Brussel,
and the International Solvay Institutes,
Pleinlaan 2, B-1050 Brussels, Belgium}
}

\date{\today}
\begin{document}
\begin{titlepage}
\maketitle
\thispagestyle{empty}
\begin{narrower}
\begin{abstract}

We define \emph{formal exponential maps} for any graded manifold as maps from the \emph{formal tangent bundle} (that we also define) into the graded manifold. We show that each such map uniquely determines and is determined by its associated \emph{Grothendieck connection}, which is shown to be flat, and to furnish a resolution of the ring of functions. 
We then show how a recent construction involving the data of a connection on the tangent bundle recovers a large class of formal exponentials in our definition. 

As an application, we use a formal exponential map to linearise a QP-manifold at a point. This gives the formal tangent space at each point the structure of an $L_\infty$-algebra with invariant inner product.

\end{abstract}
\newpage
\thispagestyle{empty}

{\hypersetup{linkcolor=black}
\tableofcontents
}\end{narrower}
\end{titlepage}

\section{Introduction}
In Riemannian geometry, the exponential map constructed from a metric on some manifold $M$ is a diffeomorphism from a neighbourhood of zero in $T_x M$ --- the tangent space at some fixed point $x \in M$ --- into a neighbourhood of $x$ in $M$. It maps the tangent vector $X \in T_x M$ to the point at affine parameter 1 along the geodesic whose tangent is $X$ at $x$. The collection of these maps ranging over $x\in M$ can be assembled into a partially defined, smooth map $\exp:TM\to M$, whose domain is an open neighbourhood $U$ of the zero section of $TM$. The merit of this construction, as far as we are concerned, is to locally identify $M$ with the tangent space $T_xM$ in a way that is smooth when the point $x$ is varied. In fact, any partially defined map $TM\to M$ satisfying a nondegeneracy condition works for this purpose, and there exist such which are not induced via geodesics \cite{Cattaneo:2001vu}.

The purpose of the current paper is to define a notion of exponential map in the context of $\mathbb Z$-graded manifolds (and, \emph{mutatis mutandis}, supermanifolds), where the Riemannian exponential map cannot be defined this way. We were motivated by a recent construction of Liao and Sti\'enon \cite{Liao:2015lta}: given a graded manifold $\cM$ equipped with a connection $\nabla$ on its tangent bundle $T\cM$, they produce what they call a \emph{formal exponential map} as a map $\mathrm{pbw}:\Gamma[ST\cM]\to D(\mathcal M)$ from sections of (graded) symmetric powers of $T\cM$ into the module of differential operators $D(\cM)$ of $\cM$. They also establish some of its properties, including the existence of its induced \emph{Grothendieck connection} on $\bar S T^\star \cM$-valued forms ($\bar S$ being the completed graded symmetric algebra functor, see formula \eqref{def:tfo}). One of our goals is to see to what extent their construction exhausts the possible formal exponential maps that can exist.

To this end we propose a general definition of a formal exponential map in section \ref{sec:formalexpmaps} as a map $$\fexp:\tfo \cM\to \cM\,,$$ where $\tfo\cM$ is the \emph{formal tangent bundle} that we also introduce there. Even when $\cM$ is an ordinary manifold, $\tfo\cM$ is a graded manifold; a kind of shifted tangent bundle, as the notation suggests. (We were inspired by a notion by the same name that sometimes appears in algebraic geometry; see \cite{pantev2013shifted} for a use in a graded symplectic context.) In Section \ref{sec:grothendieck} we show that there is a unique Grothendieck connection associated to a given formal exponential map. In the graded manifold language, this is a vector field $D$ on $T[1]\cM\oplus \tfo\cM$ that is cohomological ($D^2=0$) and has form degree 1. (This graded interpretation of the Grothendieck connection has been suggested already in \cite{Liao:2015lta} and \cite{Moshayedi:2020bkb}; realising it was another motivation for this paper.) We then show (in section \ref{sec:cohomology}) that the cohomology of $D$ is precisely the ring of functions $\cin(\cM)$ of the original graded manifold $\cM$, which is analogous to how the BV-BRST/antifield formalism of gauge theory resolves the space of gauge-invariant observables up to equations of motion \cite{Fisch:1989rp,Fisch:1989dq}. We also outline how $\fexp$ can be recovered from the data of a Grothendieck connection $D$ that provides a resolution of $\cin(\cM)$ in this way. This can be viewed as a purely algebraic characterisation of formal exponential maps in terms of homological algebra.

In section \ref{sec:examples} we exhibit examples of formal exponential maps, including a construction of a formal exponential map from the data of a connection $\nabla$ on $T\cM$. This recovers some of the results of Liao and Sti\'enon \cite{Liao:2015lta}, including their construction of a resolution of $\cin(\cM)$, with an alternative --- subjectively speaking, perhaps simpler, for some physicists --- proof. We also show that their construction recovers a large class of formal exponential maps in our definition, namely those we call \emph{proper}.

Finally in subection \ref{sec:QP} we present a straightforward application relevant for physics. If the graded manifold $\cM$ is in fact a QP-manifold (i.e.~differential graded symplectic manifold), we show that specifying a formal exponential map gives rise to a QP structure on the formal tangent space $T_x[\varepsilon]\cM$ at each point $x$ of the body of $\cM$, which varies smoothly when the point $x$ is varied. Constructing such `linearisations' was another significant motivation for this paper, and also illustrates nicely the utility of the formal tangent bundle concept, as the passage from the formal tangent bundle to the formal tangent space at $x$ is intuitive.  This QP-structure on each formal tangent space can equivalently be seen as an $L_\infty$-algebra with a cyclic inner product (of degree determined by the degree of the original QP structure), given in Kontsevich's definition in terms of formal pointed manifolds \cite{Kontsevich:1997vb}. We conclude with subection \eqref{sec:generalisedpoints} where we make some preliminary remarks on linearisation around \emph{generalised points} $x:\cS\to\cM$ (where $\cS$ is an arbitrary graded manifold instead of a point).

\subsection{Graded manifold definitions, conventions, notations}

\paragraph{Graded manifolds.} We use the same definitions as in the recent foundational monograph \cite{vysoky2021global} for graded manifolds, except where noted. A \emph{graded manifold} $\cM$ is thus a smooth real $d$-dimensional manifold $M$ --- the \emph{body} --- equipped with a sheaf of $\bbZ$-graded, $\bbZ_2$-graded-commutative algebras $\mathcal O$ so that each point of $M$ has an open neighbourhood $U$ with
\be
\label{eq:conventions:defofdmn}
\mathcal O(U)\cong \cin(U)\otimes \bar S(\bbR^m)\otimes \Lambda(\bbR^n)\,.
\ee ($\Lambda$ and $\bar S$ are respectively the exterior algebra and extended (completed) symmetric algebra, so $\bar S(\bbR^m)$ is the ring of formal power series $\bbR[[x_1,x_2,\dots x_m]]$ in commutative variables $x_1,\dots x_m$.) We assume the $\bbZ$ grading $\deg$ is correlated with the $\mathbb Z_2$ Grassmann parity, so that
\be
fg=(-1)^{\deg f \deg g} gf\equiv (-1)^{ f  g} gf
\ee
for $f,g$ lying in the \emph{ring of functions} $\cin(\cM)\equiv\mathcal O(M)$. (We just introduced a convenient notation for the Koszul sign determined by $\deg$.)

However we do \emph{not} assume the $\deg=0$ functions $\cin_0(\cM)$ are exhausted by the ring of functions $\cin(M)$ of the body $M$. (They do form a subring $\cin(M) \subseteq \cin_0(\cM)$.) This definition is therefore broader than usual definitions, which exclude degree zero ``graded'' coordinates. Similar definitions --- containing degree zero ``graded'' coordinates --- appear to have been considered in only a few works \cite{mehta2015atiyah,batakidis2018atiyah}. The point of this definition in the context of the current paper is specifically to accommodate the formal tangent bundle. Most of the usual properties true of (graded) manifolds remain true in this broader setting, although, as was pointed out in \cite{vysoky2021global}, some proofs may be more complicated; we also find a subtlety involving sections of vector bundles in this setting, that we explain at the end of subsection \ref{subsection:tforesdeg}.

We assume $M$ is connected and $\cM$ is of finite dimension (graded and otherwise).

\paragraph{Vector fields on graded manifolds.} \emph{Derivations} and \emph{vector fields} are the same and both mean left derivations $X$ on sections of the structure sheaf, so $X(fg)=X(f) g+(-1)^{Xf} f X(g)$ (where the sign is properly written $(-1)^{(\deg X)(\deg f)}$).

\paragraph{Shifted tangent bundle.}  We define the shifted tangent bundle $T[1]\cM$ in the usual way, so $\cin(T[1]\cM)=\Omega(\cM)$ (where $\Omega(\cM)$ is the ring of differential forms). The de Rham differential $\dr$ of differential forms is a degree 1 vector field $\dr$ on $T[1]\cM$, with $\dr^2=0$. $T[1]\cM$ admits an additional grading by \emph{form degree} $\operatorname{formdeg}$, where the fibre coordinates carry form degree 1.

\paragraph{Local coordinate expressions.} Without loss of generality we can assume every point of $M$ is contained in an open set $U$ that is diffeomorphic to some open subset $V$ of $\mathbb R^d$, so that $\mathcal O(U)\cong \cin(V)\otimes \bar S(\bbR^m)\otimes \Lambda(\bbR^n)$. We denote the generators of $\bar S(\bbR^m)$, $\Lambda(\bbR^n)$ and the coordinates of that $\mathbb R^d$ collectively by $\{z^a|a=1,2,\dots d+m+n\}$. These are graded commutative under multiplication:
\be
z^a z^b=(-1)^{\deg z^a \deg z^b} z^b z^a=(-1)^{ab} z^b z^a
\ee
where we again introduced a convenient notation for the sign. These coordinates are homogeneous in degree by construction.

In these coordinates we can write any vector field $X$ of $\cM$ defined on $U$ as
\be
X= X(z^a) \frac{\pd}{\pd z^a}= X^a\pd_a
\ee
using the Einstein summation convention, where $\pd_a\equiv\frac{\pd}{\pd z^a}$ is the left derivation with $\frac{\pd}{\pd z^a} z^b=\delta^b_a$.

For every such chart of $\cM$ we have a chart of $T[1]\cM$ with coordinates (where $d,m,n$ are defined in \eqref{eq:conventions:defofdmn})
\be
\label{eq:T1chart}
\{z^a|a=1,2,\dots d+m+n\}\cup\{\dr z^a|a=1,2,\dots d+m+n\}
\ee
where $\deg \dr z^a=\deg z^a +1$, $\operatorname{formdeg}z^a=0$, $\operatorname{formdeg}\dr z^a=1$, and the de Rham differential acts in the obvious way so the notation makes sense ($\dr(z^a)=\dr z^a\,, \dr (\dr z^a)=0$). Since $\dr$ is a left derivation, we have
\be
\dr f= \dr z^a\pd_a f
\ee
for $f\in \cin(\cM)$. Finally, the contraction pairing $\iota_X \alpha$ between vector fields $X$ and 1-forms $\alpha$ is such that in local coordinates
\be
\iota_X \dr f= X(f) = X^a \pd_a f\,,
\ee
so $\{\dr z^a\}$ is the dual basis to $\{\pd_a\}$.


\section{Formal exponential maps}
\label{sec:formalexpmaps}
We will define the formal exponential map as a map from the \emph{formal tangent bundle} $\tfo\cM$ to $\cM$.

\subsection{The formal tangent bundle and resolution degree}
\label{subsection:tforesdeg}
Informally $\tfo \cM$ is a ``0-shifted'' tangent bundle: one where the fibre coordinates generate a ring of formal power series, with transition functions inherited from the tangent bundle $T\cM$. ($\varepsilon$ in $\tfo$ is meant to evoke power series and/or perturbation expansions; it is not a shift by some integer in the sense of \cite[Proposition 5.25]{vysoky2021global}.) We can define it formally as
\begin{definition}
The \emph{formal tangent bundle} $\tfo\cM$ is the graded manifold with body $M$ defined by the sheaf $\mathcal O(U)=\Gamma|_U[\bar S (T^\star \cM)]$, where $U\subseteq M$ is open and $\Gamma|_U$ denotes the space of sections restricted to $U$. In particular the ring of functions of the formal tangent bundle is
\be
\label{def:tfo}
\cin(\tfo\cM)\equiv \Gamma[\bar S (T^\star \cM)]\,,\qquad\bar S (T^\star \cM)\equiv\prod_{n=0}^\infty S^n(T^\star \cM)\,.
\ee
Here $S^n$ the $n$th (graded) symmetric power and $\bar S$ is the extended (graded) symmetric algebra in the sense of \cite[Section 1.3]{vysoky2021global}. This is an algebra of (graded symmetric) formal power series on account of our use of the direct product $\prod$ instead of the direct sum.
\end{definition}

Given an arbitrary local chart of $\cM$ with coordinates $\{z^a\}$ as described above, we define the induced chart of $\tfo \cM$ to have coordinates (for $d,m,n$ counting the dimension of the body and the even and odd coordinates of $\cM$, see \eqref{eq:conventions:defofdmn})
\be
\label{adaptedtfocoords}
\{z^a|a=1,2,\dots d+m+n\}\cup\{\varepsilon^a|a=1,2,\dots d+m+n\}
\ee
where $\varepsilon^a$ are the dual basis to $\{\pd/\pd z^a\}$. We set $\deg \varepsilon^a=\deg z^a$. We declare the $\{\varepsilon^a\}$ to be the generators of the extended graded symmetric algebra in $d+m+n$ variables (see formula \eqref{def:tfo}). With the obvious definition of the transition functions (see the related formula \eqref{eq:transitionfunctions}), specifying this induced chart for each local chart of $\cM$ suffices to define $\tfo\cM$ as a graded manifold. 

A nice example is $\cM=\bbR$. Then the ring of functions consists of smooth functions on $\bbR$ tensored with formal power series in one (commuting) variable $\varepsilon$: $$\cin(\tfo\bbR)=\cin(\bbR)\otimes \bbR[[\varepsilon]]\,.$$
This illustrates that the formal tangent bundle generically (as long as the body $M$ is not a point) involves degree zero ``graded'' coordinates.

An arbitrary function $g\in \cin(\tfo \cM)$ can be expressed in these coordinates as the formal power series
\be
\label{eq:resdegexpansionG}
g=g(z,\varepsilon)= g(z,0)+\varepsilon^a\frac{\pd g}{\pd \varepsilon^a}|_{\varepsilon=0}+\frac{1}{2}\varepsilon^a \varepsilon^b \frac{\pd}{\pd\varepsilon^b}\frac{\pd g}{\pd\varepsilon^a}|_{\varepsilon=0}+\dots
\ee
We will call the polynomial degree in $\varepsilon^a$ \textbf{resolution degree}\footnote{Another good name would be ``formal degree''. Unfortunately this is too similar to ``form degree''.} and denote it by $\resdeg$. This is preserved by transition functions between local charts of this kind and therefore gives a grading on the ring of functions $\cin(\tfo\cM)$.

\paragraph{Vector bundle structure.} The formal tangent bundle $\tfo\cM$ is a graded vector bundle in the sense of \cite{vysoky2021global}. We can obtain its module of global sections as the dual module to the $\cin(\cM)$-module of functions of resolution degree 1 on $\tfo\cM$. Alternatively, the adapted coordinate chart \eqref{adaptedtfocoords} defines a local trivialisation $\tfo\cM|_U\cong \cM|_U\times \mathcal V$, for $\mathcal V$ the graded vector space whose ring of functions is the extended graded symmetric algebra (over $\bbR$) in $d+m+n$ variables $\varepsilon^a$ (see \eqref{adaptedtfocoords}). The underlying non-graded vector bundle is the trivial bundle $M\to M$ of rank zero, c.f.~$T[1]\cM$. There exists a bundle projection $p_{\tfo}:\tfo\cM \to \cM$, taking the obvious form $p_{\tfo}^\star z^a=z^a$ in the adapted chart \eqref{adaptedtfocoords}.

A subtlety with the vector bundle structure (which is in fact exhibited by any vector bundle with degree zero ``graded'' fibre coordinates) is that there is an issue with interpreting non-zero sections of $\tfo\cM$ as maps $\sfs:\cM\into\tfo\cM$ such that $p_{\tfo}\circ \sfs = 1_\cM$. This is obvious in local coordinates. A map $\sfs$ is determined by its values on the coordinates \eqref{adaptedtfocoords}, and we have $\sfs^\star z^a=z^a$ (due to $\pi_{\tfo}\circ \sfs = 1_\cM$) and $\sfs^\star\varepsilon^a\equiv \sfs^a$, for $\sfs^a$ local functions on $\cM$. However $\sfs^\star g$ will fail to be well-defined if the expansion of $g$ \eqref{eq:resdegexpansionG} in $\varepsilon$ does not terminate \emph{unless} e.g.~$\sfs^a$ vanishes. This subtlety does not affect most considerations in this paper because we will be using the zero sections of various bundles, except in the final subsection \ref{sec:generalisedpoints}.

\subsection{Vertical jets}
We return to the expansion of $g(z,\varepsilon)\in\cin(\tfo\cM)$ in terms of resolution degree. Using  the zero section $\sfs_{\tfo}:\cM\into \tfo\cM$ and the bundle projection $p_{\tfo}: \tfo\cM\to \cM$, we can write $g(z,0)$ invariantly as the pullback
$p_{\tfo}^\star \sfs_{\tfo}^\star g$. The higher-order  terms in resolution degree can also be understood invariantly as (graded) symmetric multilinear maps sending vector fields $X_1,X_2,\cdots X_n\in \Gamma[T\cM]$ to functions $f\in \cin(\cM)$.

Before we sketch that definition we explain how functions on the (ordinary) tangent bundle $T\cM$ give rise to functions on the formal tangent bundle $\tfo\cM$:
\begin{proposition}
\label{jetprop}
Given $g\in \cin(T\cM)$ a function on the tangent bundle, the following expression defines a collection $J^n(g):S^n\Gamma[T\cM]\to \cin(\cM)$ of (graded) symmetric $\cin(\cM)$-multilinear maps, that we will call the \emph{vertical jets} of $g$:
\be
\label{taylor:coordfree:def}
J^n(g)(X_1,X_2,\cdots X_n) \equiv \sfs_{T}^\star\Big(\widehat X_1\big(\widehat X_2 \cdots (\widehat X_n g)\cdots \big)\Big)\,,\qquad X_1,X_2,\cdots \in \Gamma[T\cM]\,,
\ee
for $\sfs_{T}:\cM\to T\cM$ the zero section of $T\cM$, and $\widehat X \in \Gamma[T(T\cM)]$ the \emph{vertical lift} of $X\in \Gamma[T\cM]$ (to be defined). Moreover, $J^n(g)$ only depends on the equivalence class
\be
\label{taylor:equivalencerel}
g_1\sim g_2\iff g_1-g_2 \in \mathcal I^{n+1}
\ee
for $\mathcal I^{n+1}$ the $(n+1)$-th power of the ideal $\mathcal I$ of functions in $\cin(T\cM)$ that are annihilated by $\sfs^\star_T$ (i.e.~those that vanish on the locus $\cM\into T\cM$ defined by the zero section).
\end{proposition}
\begin{proof}
The vertical lift is defined by \cite{yano1964linear}
\be
\label{verticalliftdef}
\widehat X(\alpha)=p_{T}^\star(\iota_X\alpha)\,,\qquad \widehat X\pi^\star_T =0\,,
\ee
where $\alpha$ on the left-hand side is the function $\alpha \in \cin(T\cM)$ corresponding to the 1-form $\alpha\in \Gamma[T^\star\cM]$. In adapted local coordinates $\{z^a,v^a\}$ for $T\cM$ (analogous to those of \eqref{adaptedtfocoords} for $\tfo\cM$), where $v^a$ are fibre coordinates, we can write the 1-form as $\alpha(z,v)= v^a \alpha_a(z)$ and the vector as $X=X^a(z)\pd/\pd z^a$. We then calculate
$\widehat X= X^a(z){\pd}/{\pd v^a}
$,
whence the graded commutator of two vertical lifts vanishes: $[\widehat X_1,\widehat X_2]=0$. Furthermore we have a kind of $\cin(\cM)$-linearity:
\be
\widehat{f X}= p_T^\star(f) \widehat X\,.
\ee
Since $\sfs_T^\star p_T^\star=1^\star_\cM$ this implies linearity:
\be
\begin{aligned}
 J^n(g)(f X_1,X_2,\cdots X_n) =\sfs^\star\Big(\widehat{f X_1}\big(\widehat X_2 \cdots (\widehat X_n g)\cdots \big)\Big)&=\sfs_T^\star\Big(p_T^\star f\widehat{X}_1\big(\widehat X_2 \cdots (\widehat X_n g)\cdots \big)\Big)\\
 &=f \sfs_T^\star\Big(\widehat{X}_1\big(\widehat X_2 \cdots (\widehat X_n g)\cdots \big)\Big)\,.
\end{aligned}
 \ee
 Finally, if $h\in \mathcal I^{n+1}$, then by the Leibniz rule $\widehat X_1\big(\widehat X_2 \cdots (\widehat X_n h)\cdots \big)$ is in $\mathcal I$, so $J^n(h)=0$.
\end{proof}

The same definition \eqref{taylor:coordfree:def} also works for functions $g$ defined on any open set $U$ containing the zero section $\sfs_T(M)$. In fact, if $g_1$ agrees with $g_2$ in some such $U$ up to terms in $(\mathcal I_U)^{n+1}$ (for $\mathcal I_U$ the ideal of functions defined on $U$ that are annihilated by $\sfs^\star|_{U}$), then $J^n(g_1)=J^n(g_2)$. Proposition \ref{jetprop} therefore gives a coordinate-free definition of ``jets along fibre directions'' of a function $g$ defined in a neighbourhood of the zero section of $T\cM$ that works in the graded case. In local coordinates the vertical jets $J^n(g)$ are recognised as the Taylor coefficients in an expansion around the zero section (given here assuming $\deg X_i=0\mod 2$):
\be
J^n(g)(X_1,X_2,\cdots X_n)=X_1^{a_1} X_2^{a_2} \cdots X_n^{a_n} \frac{\pd^n}{\pd v^{a_n}\pd v^{a_{n-1}}\cdots \pd v^{a_1}}g\Big|_{v=0}\,.
\ee

Since we can view 1-forms on $\cM$ as resolution degree 1 elements of $\cin(\tfo\cM)$, we can mimic the vertical lift \eqref{verticalliftdef} and jet \eqref{taylor:coordfree:def} constructions for $\tfo$ instead of $T$. To every $g\in\cin(\tfo \cM)$ this associates a collection $J_{\tfo}^n(g)\,, n=0,1,\dots$ of (graded) symmetric multilinear maps. These uniquely define $g$:
\be
\cin(\tfo\cM)\ni g=0\iff J_{\tfo}^n(g)=0\,,\qquad \forall n=0,1,\dots\,.
\ee
Given some $g$ defined (locally) on $T\cM$ we can therefore obtain a $g_\tfo$ defined (globally) on $\tfo\cM$ if we set $J^n_{\tfo}(g_{\tfo})\equiv J^n(g)$. (Henceforth we do not distinguish between $J^n$ and $J^n_{\tfo}$.)

Putting everything together:
\begin{proposition}
\label{preFexpfromgenexpmap:prop}
For every pair $(U,g)$ of an open set $U\supset \sfs_T(M)$ containing the zero section and function $g$ on $T\cM$ defined locally on $U$,  there exists a unique $g_{\tfo}$ defined (globally) on $\tfo \cM$ via its vertical jets \eqref{taylor:coordfree:def}. The function $g_{\tfo}$ depends only on the equivalence class
\be
\begin{split}
(U_1,g_1)\sim (U_2,g_2) \iff \exists V\subseteq U_1\cap U_2\,, V\supset \sfs_T(M) \,,  \\ (g_1-g_2)|_{V}\in (\mathcal I_V)^{n+1} \qquad\forall n=0,1,\dots
\end{split}
\ee
for $V$ an open set and $\mathcal I_V$ the ideal of functions on $T\cM$ defined locally on $V$ that are annihilated by the pullback $\sfs_T^\star$ involving the zero section $\sfs_T:\cM\to T\cM$.
\end{proposition}
\noindent In adapted local coordinates this map is realised by the replacement $v^a\to \varepsilon^a$ inside the Taylor expansion of $g(z,v)$ around $v=0$.

\subsection{Definition of formal exponential maps}
\begin{definition}
\label{fexpdef}
A map (of graded manifolds) $\fexp:\tfo\cM\to \cM$ is a \emph{formal exponential map} if and only if it satisfies the two properties
\begin{subequations}
\begin{align}
   J^0(\fexp^\star f)=f \qquad \forall f\in \cin(\cM)\,, \label{eq:fexp:def:propertyA}
 \\
   \text{if } f\in \cin(\cM)\,,\qquad \dr f=0 \iff J^1(\fexp^\star f)(X)=0 \quad \forall X\in \Gamma[TX]   \label{eq:fexp:def:propertyB}\,.
\end{align}
\end{subequations}
We will call $\fexp$ a \emph{proper} formal exponential map in the special case
\be
\label{eq:fexp:def:proper}
J^1(\fexp^\star f)(X)=X(f)\,.
\ee
(The first property \eqref{eq:fexp:def:propertyA} can also be stated more transparently as  $f=\sfs_{\tfo}^\star \fexp^\star f\; \forall f$ or as $\fexp\circ \sfs_{\tfo}=1_{\cM}$, for $\sfs_{\tfo}$ the zero section $\cM\into \tfo\cM$.)
\end{definition}

These properties serve to characterise the first two terms in the expansion of $\fexp^\star f$ in resolution degree; we have
\be
\label{fexpTaylor}
(\fexp^\star f)(z,\varepsilon)= f(z) + \varepsilon^a e_a^b(z)\frac{\pd f}{\pd z^b}|_{\pd\varepsilon=0} +\dots
\ee
where $e_a^b(z)$ is required to admit an inverse by property \eqref{eq:fexp:def:propertyB}:
\be
e_a^b(z)c_b^c(z)=\delta^c_a
\ee
for some $c_b^c$ depending on $z$. For a \emph{proper} formal exponential map $e_a^b(z)$ is the unit matrix: $e_a^b(z)=\delta_a^b$. We can equivalently characterise a formal exponential map by demanding its expansion takes the form \eqref{fexpTaylor} everywhere, which will be convenient later.


\section{The Grothendieck connection of a formal exponential map}
\label{sec:grothendieck}
\subsection{Uniqueness}
The relation between the formal exponential map $\fexp$ and associated flat connection $D$ can be understood in terms of a simple diagram of \emph{graded manifolds with distinguished vector fields}. A morphism of such is a map $\varphi:\cM\to \mathcal N$ of graded manifolds that relates the distinguished vector fields $V_{\cM}, V_{\mathcal N}$:
\be
\varphi^\star V_{\mathcal N}= V_{\mathcal M} \varphi^\star\,.
\ee
If $V_{\mathcal M}$ and $V_{\mathcal N}$ both square to zero, this is the more familiar notion of a \emph{dg-map}, also known as a \emph{morphism of Q-manifolds}.

The diagram in question is
\be
\label{DiagramThatsDg}
\begin{tikzcd}[every label/.append style = {font = \small}]
        &    T[1]\cM\oplus \tfo\cM \ar[d, "\pi_{T[1]}"',two heads]\ar[ddr,"\mathsf{e}",two heads,bend left=30]&  \\
        &    T[1]\cM                                                                                                            & \\
 \cM\ar[uur,"\mathsf{s}_0",hook,bend left=30]\ar[rr,"1_\cM"]&                                                                               & \cM  
\end{tikzcd}
\ee
Here $\pi_{T[1]}$ is the vector bundle projection, $\mathsf{s}_0$ is the zero section of $T[1]\cM\oplus \tfo \cM$, and $\sfe$ is identified via (what will turn out to be) the formal exponential map:
\be
\label{efact}
\sfe=\fexp\circ \pi_{\tfo}\,,
\ee
where $\pi_{\tfo}:T[1]\cM\oplus \tfo \cM\onto \tfo \cM$ is the bundle projection.

\begin{proposition}\label{uniqueness:prop}If there exist
\begin{itemize}
\item a map $e:T[1]\cM\oplus \tfo\cM\to \cM$ of graded manifolds that factors as \eqref{efact} for some $\fexp:\tfo\cM\to \cM$,
\item and a vector field $D$ on $T[1]\cM\oplus \tfo\cM$ with form degree $1$
\end{itemize}
 so that \eqref{DiagramThatsDg} is a commutative diagram of graded manifolds with distinguished vector fields (those being $D$, $\dr$, and $0$ on $T[1]\cM\oplus \tfo \cM$, $T[1]\cM$, and $\cM$ respectively), then
 \begin{enumerate}
    \item $\fexp$ is a formal exponential map,
    \item $D$ is uniquely determined in terms of $\fexp$ and vice versa.
 \end{enumerate}
\end{proposition}
\begin{definition}
If a pair $(\fexp,D)$ exist satisfying the assumptions of Proposition \ref{uniqueness:prop} we will call $D$ the \emph{Grothendieck connection} uniquely associated to the formal exponential map $\fexp$.
\end{definition}
The diagrammatic assumptions read, in terms of formulas, \begin{align}
\sfe\circ\sfs_0&=1_{\cM}\,, \label{sfeassumption}\\
D\pi^\star_{T[1]}&=\pi^\star_{T[1]}\dr\,, \label{DvsdeRhamcondition}\\
 D\sfe^\star&=0 \,.\label{Dkillsecondition}
\end{align}

\begin{proof}[Proof (Proposition \ref{uniqueness:prop}).]
It suffices to consider $\sfe$ and $D$ in local coordinates. Condition \eqref{DvsdeRhamcondition} along with $\operatorname{formdeg}D=1$ imply $D$ takes the following form:
\be
D=\dr z^a \frac{\pd}{\pd z^a} + \dr z^a C_a^b(z,\varepsilon)\frac{\pd}{\pd \varepsilon^b}\,.
\ee
The matrix $C_a^b(z,\varepsilon)$ is a formal power series in $\varepsilon$. We expand it as 
\be
C_a^b(z,\varepsilon)={C_{-1}}^b_a(z) + \varepsilon^c {C_{0}}_{ac}^b(z) +\dots + \frac{1}{n!}(\varepsilon^{c_1}\cdots \varepsilon^{c_{n+1}}) {C_{n}}_{ac_1\dots c_n}^b(z)+\dots\,,
\ee
and consider the corresponding expansion of $D$:
\be
\label{Dexpansionearly}
D=\dr + \sum_{n=-1}^\infty C_n\,,\qquad \Big(C_n\equiv \dr z^a\frac{1}{n!}(\varepsilon^{c_1}\cdots \varepsilon^{c_{n+1}}) {C_{n}}_{ac_1\dots c_n}^b(z)\frac{\pd}{\pd \varepsilon^b}\Big)\,.
\ee
$C_n$ increases the resolution degree (polynomial degree in $\varepsilon$) by $n$.

The local coordinate expression of $\sfe$ is also an expansion in resolution degree:
\be
\label{sfeexpansion}
\sfe^\star(z^a)= z^a+\varepsilon^b {e_1}_b^a(z) +\dots+\frac{1}{n!} (\varepsilon^{b_1}\cdots \varepsilon^{b_n}) {e_n}^a_{b_1\dots b_n}(z)+\dots=z^a + e_1^a(z,\varepsilon) +\dots
\ee
We already used condition \eqref{sfeassumption} and $\sfe=\fexp\circ\pi_{\tfo}$ to deduce this expression.
The matrix ${e_1}_b^a(z)$ is essentially the differential of $\fexp$ at the zero section along the formal directions, and $e_{n>1}$ correspond to its higher jets.

We now organise condition \eqref{Dkillsecondition} in resolution degree $\ell$:
\be
D\sfe^\star=0\iff \begin{cases} \dr z^a + C_{-1} e^a_1=0  &(\ell=0)\\
\sum_{n=-1}^{\ell-1} C_n e^a_{\ell-n}=0  &(\ell>0)
\end{cases}
\ee
The $\ell=0$ condition is equivalent to ${C_{-1}}_c^b(z) {e_1}_b^a(z)=-\delta^a_c$, so ${e_1}_b^a$ is invertible. This implies $\fexp$ takes the form 
\be
\fexp^\star(z^a)=z^a+\varepsilon^b {e_1}_b^a(z) +\dots
\ee
for invertible $e_b^a$. Therefore $\fexp$ satisfies the two assumptions \eqref{eq:fexp:def:propertyA}, \eqref{eq:fexp:def:propertyB} defining a formal exponential map.

We can rewrite the $\ell>0$ conditions as
\be
C_{-1} e^a_{\ell+1}+ C_{\ell-1} e_1^a+\sum_{n=0}^{\ell-2} C_n e^a_{\ell-n}=0\,.
\ee
Given knowledge of the formal exponential map $\fexp$, the expressions $e^a_n$ are known for all $n$. Since ${e_1}_b^a$ is invertible, this formula determines $C_{\ell-1}$ recursively in terms of $C_{n<\ell-1}$. Conversely, if $D$ is known, the formula determines $e_{\ell+1}^a$ recursively in terms of $e_{1},e_{2},\dots e_{\ell}$: since
\be
C_{-1} e^a_{\ell+1}= \frac{1}{\ell!}(\dr z^b)\;{C_{-1}}^c_b\; \varepsilon^{d_1}\cdots \varepsilon^{d_\ell}\,{(e_{\ell+1})}^a_{cd_1\dots d_\ell}
\ee
and ${C_{-1}}_b^a$ is invertible, we can solve for ${(e_{\ell+1})}^a_{cd_1\dots d_\ell}(z)$. Thus, $D$ and $\fexp$ uniquely determine each other.
\end{proof}

\subsection{Flatness}
\begin{proposition}
\label{flatnessprop}
Under the assumptions of Proposition \ref{uniqueness:prop}, $D^2=0$ so in fact diagram \eqref{DiagramThatsDg} is a diagram of Q-manifolds, and $D$ is a flat connection.
\end{proposition}
\begin{proof}
The idea is to use an invertible map $\rho: T[1]\cM\oplus \tfo \cM\to T[1]\cM\oplus \tfo \cM$ so that $\rho^\star e^\star(z^a)=z^a+\varepsilon^a$. This $\rho$ only needs to exist in some neighbourhood of an arbitrary point of $M$, because $D^2=0$ is a local statement. (We will in fact define it in a single local coordinate patch.) If $\rho\circ \sfs_0=\sfs_0$ as well, Proposition \ref{uniqueness:prop} determines the unique $D'\equiv \rho^\star D(\rho\inv)^\star$ corresponding to $\sfe'\equiv \sfe\circ \rho$, which takes the form
\be
D'=\dr -\dr z^a \frac{\partial}{\pd \varepsilon^a}\,.
\ee
This squares to zero, and so $D$ will also square to zero.

We first define $\rho^\star(z^a)=z^a\,,\rho^\star(\dr z^a)=\dr z^a$. We consider the case where ${e_1}_b^a=\delta^a_b$ i.e.~that of a proper formal exponential map for simplicity. Then we define $\rho^\star$ in another expansion in resolution degree:
\be
\rho^\star(\varepsilon^a)=\varepsilon^a +\frac{1}{2} \varepsilon^{b_1}\varepsilon^{b_2} {\rho_2}^a_{b_1b_2}(z)+\dots\,.
\ee
For this ansatz we have $\rho\circ \sfs_0=\sfs_0$ as desired. Now $\rho^\star \sfe^\star z^a=z^a+\varepsilon^a$ can be solved  order-by-order in resolution degree for the coefficients ${\rho_n}^a_{b_1\dots b_n}(z)$  in terms of the coefficients ${e_m}^a_{b_1\dots b_n}(z)$ ($m\leq n$) defining the local expression for the exponential map as well as the $\rho_{n'}$ ($n'< n$).
\end{proof}

\subsection{Transferring to a diffeomorphic manifold}
\label{sec:diffeo}
Given a formal exponential map along with its Grothendieck connection $(\fexp,D)$ on some graded manifold $\cM$, we will obtain a pair $(\overline{\fexp},\overline{D})$ on a diffeomorphic manifold $\overline{\cM}$, related to $\cM$ via a diffeomorphism
\be
\varphi:\overline{\cM}\to\cM\,.
\ee
This is essentially trivial but will be useful for constructing $D$ from $\fexp$ shortly.

We will need notation for the differential of $\varphi$ relevant for each kind of tangent bundle. We will call $T[1]\varphi,\tfo\varphi, T_\mathsf{big}\varphi$ respectively the maps $T[1]\overline{\cM}\to T[1]\cM$, $\tfo\overline{\cM}\to \tfo\cM$, and $T[1]\overline\cM\oplus \tfo\overline\cM\to T[1]\cM\oplus \tfo\cM$, all induced by $\varphi$. The last two fit into the commutative diagram
\be
\begin{tikzcd}
T[1]\overline{\cM}\oplus \tfo\overline{\cM} \ar[r,"T_{\mathrm{big}}\varphi"] \ar[d,"\overline{\pi_{\tfo}}"]           &T[1]\cM\oplus \tfo\cM \ar[d,"\pi_{\tfo}"] \\
\tfo\overline{\cM}\ar[r,"\tfo\varphi"] & \tfo\cM
\end{tikzcd}
\ee
whence the identity
\be
\tfo\varphi \circ \overline{\pi_{\tfo}}= \pi_{\tfo}\circ T_{\mathrm{big}}\varphi\,. \label{eq:derivativeidentity}
\ee

Since $D$ is a vector field on $T_\mathrm{big}\cM$ it will pullback to $\overline D\equiv\big(T_\mathrm{big}\varphi\big)^\star D \big(T_\mathrm{big} (\varphi^{-1})\big)^\star$. For $\fexp$ we instead have $\overline{\fexp}\equiv\varphi^{-1}\circ\fexp\circ \tfo\varphi$ which is clearly a formal exponential map $\tfo\overline\cM\to \overline \cM$ whenever $\fexp$ is. Showing that $\overline{D}$ is the Grothendieck connection associated to $\overline{\fexp}$ requires a short calculation.

For $\overline e\equiv \overline{\fexp}\circ\overline{\pi_{\tfo}}$, identity \eqref{eq:derivativeidentity} implies
\begin{align}
\overline e=\varphi^{-1}\circ\fexp\circ \tfo\varphi \circ\overline{\pi_{\tfo}}&=\varphi^{-1}\circ\fexp \circ \pi_{\tfo}\circ T_{\mathrm{big}}\varphi\\
&=\varphi^{-1} \circ \sfe \circ T_\mathrm{big}\varphi
\end{align}
Therefore we find that condition \eqref{Dkillsecondition} that determines the Grothendieck connection in terms of the exponential map is satisfied:
\be
\overline{D}\overline{\sfe}^\star=\big(T_\mathrm{big}\varphi\big)^\star D \big(T_\mathrm{big} (\varphi^{-1})\big)^\star \Big( (T_\mathrm{big}\varphi)^\star \sfe^\star (\varphi^{-1})^\star \Big)=\big(T_\mathrm{big}\varphi\big)^\star D \sfe^\star  (\varphi^{-1})^\star =0\,.
\ee

The pullback $(T_\mathrm{big}\varphi)^\star$ preserves form degree. This is even more obvious from the coordinate expression, that we list here to be totally explicit:
\be
\label{eq:transitionfunctions}
(T_\mathrm{big}\varphi)^\star z^a=\varphi^\star(z^a)\,,\quad (T_\mathrm{big}\varphi)^\star (\dr z^a)= \dr \overline{z}^b \frac{\pd \varphi^\star z^a}{\pd \overline{z}^b}\,,\quad (T_\mathrm{big}\varphi)^\star\varepsilon^a= \overline{\varepsilon}^b \frac{\pd \varphi^\star z^a}{\pd \overline z^b}\,.
\ee
Therefore the assumptions of Proposition \ref{uniqueness:prop} are satisfied and we have
\begin{proposition}
\label{diffprop}
If $(\fexp,D)$ are a pair of a formal exponential map and its Grothendieck connection on a graded manifold $\cM$, and $\varphi:\overline{\cM}\to \cM$ is a diffeomorphism, then $(\overline{\fexp},\overline{D})$ are another such pair on $\overline{\cM}$, explicitly given by
\be
\label{diffformula}
\overline{\fexp}\equiv\varphi^{-1}\circ\fexp\circ \tfo\varphi\,,\qquad  \overline D\equiv\big(T_\mathrm{big}\varphi\big)^\star D \big(T_\mathrm{big} (\varphi^{-1})\big)^\star\,.
\ee
\end{proposition}

\subsection{Existence}
In the proof of Proposition \ref{uniqueness:prop} we explicitly construct a Grothendieck connection $D$ given a formal exponential map $\fexp$ on any local coordinate chart of $\cM$; it remains to argue that these patch together to globally-defined vector fields on $T[1]\cM\oplus \tfo\cM$.

We cover the body $M$ with a collection $\{U_\alpha\}$ of open sets that each admit a local coordinate chart of $\cM$ and consider the associated transition functions $\varphi_{\alpha\beta}$, which are diffeomorphisms on overlaps $U_\alpha\cap U_\beta$. This atlas induces local coordinate charts of the tangent bundle $T\cM$ and the related bundles $T[1]\cM,\tfo \cM,\dots$ that interest us; their transition functions will be given by the differentials $T[1]\varphi_{\alpha\beta},\tfo\varphi_{\alpha\beta},\dots$.

We denote $\fexp_{\alpha}$  the representative of $\fexp:\tfo\cM\to \cM$ associated to the local charts of $\cM$ and $\tfo \cM$ respectively on $U_\alpha$ and $\fexp_\beta$ the same on $U_\beta$. Since $\fexp$ is a globally-defined map $\tfo \cM\to \cM$, $\fexp_\alpha|_{U_\alpha \cap U_\beta}$ and $\fexp_\beta|_{U_\alpha \cap U_\beta}$ are related to each other as in \eqref{diffformula} (for $\varphi=\varphi_{\alpha\beta}$); therefore, the locally defined vector fields $D_\alpha$, $D_\beta$  are also related to each other as in \eqref{diffformula}. But this is precisely how the local representatives of a globally-defined vector field on $T[1]\cM\oplus \tfo \cM$ are related to each other with this choice of atlas. This sketches a \emph{construction} of $D$ from the data of $\fexp$:
\begin{theorem}
\label{Dfromfexp:thm}
If $\fexp$ is a  formal exponential map, there \emph{exists} a (unique) Grothendieck connection associated to it.
\end{theorem}

\section{Cohomology of the Grothendieck connection}
\label{sec:cohomology}
We will calculate the cohomology of $D$, in steps. This cohomology will turn out to be isomorphic to the ring of functions $\cin(\cM)$; thus the complex $\cin(T[1]\cM\oplus \tfo\cM)$ equipped with the differential $D$ will provide a \emph{resolution} of $\cin(\cM)$. (This motivated the name ``resolution degree''.) 

\subsection{Acyclicity of $D_{-1}$}
We start from the expansion of $D$ in resolution degree. This is in fact globally well-defined, as the transition functions of $T[1]\cM\oplus \tfo\cM$ can be chosen to preserve resolution degree, see \eqref{eq:transitionfunctions}. (This exploits the fact $T[1]\cM\oplus \tfo\cM$ is not just a graded manifold, but a graded vector bundle.) Said expansion is
\be
\label{Dexpansion}
D=D_{-1}+ D_0 + D_{1}+\dots
\ee
This is almost the same as the expansion \eqref{Dexpansionearly} in local coordinates, with $D_n=C_n$ for $n\neq 0$ and $D_0=\dr + C_0$ for $n=0$ (for $\dr=\dr z^a \pd_a$).

The flatness condition $D^2=0$ gives rise to an infinite number of equations indexed by resolution degree, the first three of which are
\be
\label{expandedflatness}
D_{-1}^2=0\,,\qquad  D_{-1}D_0 +  D_0 D_{-1}=0\,,\qquad D_0^2=-(D_{-1}D_{+1}+ D_{+1} D_{-1})\,,\qquad \dots
\ee

Therefore $D_{-1}$ squares to zero. We will show $D_{-1}$ is \emph{acyclic}, i.e.~its homology is concentrated in $\deg_{\mathrm{HT}}=0$, where this new degree is defined as\be
\deg_{\mathrm{HT}}\equiv \operatorname{resdeg}+\operatorname{formdeg}\,.
\ee
(The subscript HT is motivated in appendix \ref{app:HT}.) For proper formal exponential maps which arise from a choice of connection on $T\cM$ this is proven by Liao-Sti\'{e}non \cite{Liao:2015lta} using a technique by Dolgushev \cite{dolgushev2005covariant}; we will give a short argument that also works for not-necessarily-proper formal exponential maps, using essentially the same tools.

We invoke a vector field $\zeta$ on $T[1]\oplus T[\varepsilon]$ which is required to have (for $\delta\equiv D_{-1}$) 
\be
\label{epsilonHT}
\deg \zeta=1\,,\qquad\operatorname{resdeg}\zeta=1\,,\qquad \operatorname{formdeg} \zeta=-1\,,\qquad
\zeta \delta+\delta \zeta= \dr z^a \frac{\pd}{\pd \dr z^a}+ \varepsilon^a \frac{\pd}{\pd \varepsilon^a}\,.
\ee
The last expression is the (globally-defined) vector field $\epsilon_{\mathrm{HT}}$ that counts $\deg_{\mathrm{HT}}$.

The vector field $\zeta$ can be constructed locally as 
\be
\zeta= \varepsilon^b e_b^a(z) \frac{\pd}{\pd \dr z^a}
\ee
for $e_b^a(z)$ as in \eqref{sfeexpansion} the inverse of $C_a^b(z)$ up to a sign. This can be glued consistently as in the proof to Theorem \ref{Dfromfexp:thm}.

We then introduce the deformation retract $h$ of $T[1]\oplus T[\varepsilon]$ to its zero section that is a global linear scaling of each fibre to zero. Explicitly, in local coordinates:
\be
\label{deformationretractdef}
\begin{split}
   h:[0,1]\times (T[1]\oplus T[\varepsilon])\to T[1]\oplus T[\varepsilon]\,,\\
   h^\star z^a=z^a\,,\quad h^\star \dr z^a=t \dr z^a\,,\quad h^\star \varepsilon^a=t \varepsilon^a\,, \qquad t\in[0,1]\,.
\end{split}
\ee
Degree-counting implies $\zeta$ and $\delta$ commute with $h^\star$. Then a suitable homotopy is
\be
\label{homotopydef}
H f\equiv \zeta \int_0^1 (h^\star f)(t) \frac{\dr t}{t}\,,\qquad H:\cin(T[1]\oplus T[\varepsilon])\to \cin(T[1]\oplus T[\varepsilon])\,,
\ee
which is well-defined and satisfies
\be
\label{homotopyidentity}
(H\delta +\delta H)f=\begin{cases}  f &(\deg_{\mathrm{HT}}f\neq 0)\,,\\
                                    0 &(\deg_{\mathrm{HT}}f= 0)\,.
                     \end{cases}
\ee
This proves
\begin{proposition}
\label{acyclicityOfDelta:prop}
If $D$ is the Grothendieck connection of a formal exponential map, then $H_0(D_{-1})\cong \cin(\cM)$, $H_{k}(D_{-1})=0$ ($k>0$) where $k$ indexes $\deg_{\mathrm{HT}}$.
\end{proposition}



\subsection{Calculation of the cohomology}
At this point the cohomology of $D$ can be determined by invoking an abstract homological perturbation result. However we calculate it explicitly here in order to argue, later, that the same calculation works when we replace $\cM$ by its restriction to some open $U\subseteq M$ of the body $M$. (Also, this specific case seems simpler to prove than the general homological perturbation scenario.) In appendix \ref{app:HT} we make the analogy with the homological interpretation of the BV-BRST/antifield formalism explicit.

For a function $f\in \cin(T[1]\cM\oplus \tfo \cM)$ we will say $\resdeg f\geq n$ if all terms in its expansion in $\resdeg$ vanish below $\resdeg=n$. 
\begin{lemma}
\label{cohomologylemma1}
If $Df=0$ and $f$ is cohomologous to $g$ of $\formdeg$ $k$ with $\resdeg g\geq n>0$, then $f$ is cohomologous to $g'$ of $\formdeg$ $k$ with $\resdeg g'\geq n+1$ (so in fact $f$ is trivial in cohomology).
\end{lemma}
\begin{proof}
We have $g=g_n+g_{n+1}+\cdots$ (subscripts denote $\operatorname{resdeg}$). Thus $Dg=0$ implies $D_{-1}g_n=0$, which gives $g_n=D_{-1} H g_n=D H g_n - (D-D_{-1})H g_n$ via \eqref{homotopyidentity}. Therefore
\be
g=DHg_n - (D-D_{-1})H g_n + g_{n+1}+\cdots
\ee
The term $(D-D_{-1})H g_n$ has $\resdeg \geq (n+1)$ because $H$ increases $\resdeg$ while $(D-D_{-1})$ does not decrease it. Thus $g'=g-DHg_n$ has $\resdeg\geq (n+1)$ and the same $\formdeg$.
\end{proof}
\begin{lemma}
If $Df=0$, $\resdeg f=0$, and $\formdeg f >0$, then $f$ is trivial in cohomology.
\end{lemma}
\begin{proof}
$\resdeg f=0$ implies $D_{-1}f=0$. We then have $f=DHf- (D-D_{-1}) Hf$ via \eqref{homotopyidentity}. The last term has $\resdeg\geq 1$, so $f$ is trivial by Lemma \ref{cohomologylemma1}.
\end{proof}
\begin{lemma}
\label{cohomologylemma3}
For all $g\in\cin(\cM)$ there exists a \emph{unique} $f\in \cin(T[1]\cM\oplus \tfo\cM)$ with $Df=0$, $\formdeg f=0$, and $\sfs_0^\star f=g$ (for $\sfs_0$ the zero section of $T[1]\cM\oplus \tfo\cM$).
\end{lemma}
\begin{proof}
\emph{Uniqueness:} if $f,f'$ are two such, $f-f'$ has $\sfs_0^\star(f-f')=0\implies \resdeg (f-f')\geq 1$. By Lemma \ref{cohomologylemma1} this difference is cohomologous to a trivial 0-form; this is zero.

\emph{Existence:} We expand the putative $f$ in resolution degree as $f=f_0+f_1+\dots$, where $f_0$ is the pullback of $g$ by the bundle projection. $D^2=0$ implies $D_0 D_{-1}+D_{-1} D_0=0$, which implies $D_{-1}D_0 f_0=0$. If we set $f_1=- H D_0 f_0$ we thus find $\resdeg D(f_0+f_1)\geq 1$. We now employ Lemma \ref{cohomologylemma1} to inductively show there exist $f_1,f_2,\dots f_\ell$ so that
$   \resdeg D(f_0+f_1+\dots f_\ell)\geq \ell
   $
   for all $\ell$.
\end{proof}
Uniqueness implies the map $g\to f$ of Lemma \ref{cohomologylemma3} is multiplicative. By Proposition \ref{uniqueness:prop} this map is the pullback $\sfe^\star$ defined in \eqref{efact} via $\fexp$. It is related to the map $\tau$ of Liao-Sti\'{e}non defined via their \cite[Proposition 6.2]{Liao:2015lta} which is essentially our Lemma \ref{cohomologylemma3}.

The three lemmas together calculate the cohomology $H^\bullet(D)$:
\begin{theorem}
If $D$ is a Grothendieck connection, then $H^{k>0}(D)=0$, $H^0(D)\cong \cin(\cM)$, where $k$ is form degree. The map $\sfe^\star: \cin(\cM)\to \cin(T[1]\cM\oplus \tfo \cM)$ is a quasi-isomorphism.
\end{theorem}

\subsection{From the Grothendieck connection to the formal exponential map}
For this subsection we do not assume $D$ is a Grothendieck connection; we assume it is just a vector field on $T[1]\cM\oplus \tfo \cM$ that lifts the de Rham differential (formula \eqref{DvsdeRhamcondition}, i.e.~$D\pi_{T[1]}^\star=\pi_{T[1]}^\star \dr$) and has $\formdeg D=1$. $D_{-1}$ will again be the vector field with $\resdeg 1$ in the expansion of $D$ in resolution degree.

We can recover the formal exponential map from $D$ alone:
\begin{theorem}
If $D$ as above has $D^2=0$, and if $D_{-1}$ is acyclic in the sense of Proposition \ref{acyclicityOfDelta:prop}, then there exists a unique formal exponential map $\fexp :\tfo \cM\to \cM$ for which $D$ is the Grothendieck connection in the sense of Proposition \ref{uniqueness:prop}.
\end{theorem}
\begin{proof}(Sketch.)
The arguments of the previous two subsections are valid if we simultaneously replace $\cM$ by its restriction $\cM|_U$ on some open set $U\subseteq M$ and all vector fields by their restrictions to $U$. The deformation retract \eqref{deformationretractdef} fixes $\cM$ so in particular it fixes the body $M$, so it and the homotopy $H$ of \eqref{homotopydef} are also well-defined on the restriction. Lemma \ref{cohomologylemma3} then defines a map $(\sfe|_U)^\star: \cin(\cM|_U)\to \cin\big( (T[1]\oplus \tfo)\cM|_U\big)$. It can be checked that the collection of such for each $U$ defines a morphism of graded locally ringed spaces in the sense of \cite[Definition 2.24]{vysoky2021global} which is to say a morphism of graded manifolds $\sfe:(T[1]\oplus \tfo)\cM \to \cM$. This factors as $\sfe=\fexp\circ \pi_{\tfo}$ for some map $\fexp:\tfo\cM\to \cM$ by degree-counting reasons, and so Proposition \ref{uniqueness:prop} then implies $\fexp$ is a  formal exponential map.
\end{proof}

\section{Examples and constructions}
\label{sec:examples}
\subsection{The canonical exponential on a vector space} Assume $\cM=$ is a \emph{linear} graded manifold $\mathcal V$ i.e.~a vector space. In that case it can be covered with a single global coordinate chart, and we can define
\be
\fexp^\star z^a\equiv z^a+ \varepsilon^a\,.
\ee
This arises from the exponential map $\exp:T\cM\to \cM$ defined by the obvious expression $\exp^\star z^a=z^a + v^a$. This is sensible since $T\cM\cong {\mathcal V}\oplus {\mathcal V}$, and it is globally defined. Any formal exponential map on a linear graded manifold can be transformed into this canonical one using a ``formal diffeomorphism'' acting only on $\{\varepsilon^a\}$, as outlined in the proof of Proposition \ref{flatnessprop}. (This does not preserve resolution degree.)

\subsection{Formal exponentials from generalised exponentials}
\label{sec:fexpfromgexp}
Using Proposition \ref{preFexpfromgenexpmap:prop} we can construct formal exponential maps from the vertical jets of certain maps
\be
\exp:T\cM\to \cM\,,
\ee
locally defined on some open $U \subseteq TM$ containing the zero section of $TM$. (More precisely, from $J^n(\exp^\star f)$ where $f$ ranges over $\cin(\cM)$.) When the properties analogous to (\ref{eq:fexp:def:propertyA},~\ref{eq:fexp:def:propertyB}) are satisfied, we say $\exp$ is an \emph{exponential map}. \emph{Proper} exponential maps in the non-graded case have been called ``generalised exponential maps'', see e.g.~\cite{Cattaneo:2001vu}. Usually --- e.g.~for the Riemannian exponential map --- condition \eqref{eq:fexp:def:proper} is satisfied so $\exp$ is a generalised exponential map and its associated $\fexp$ is a proper formal exponential map.

\subsection{The formal exponential associated to a connection}\label{fexpconnsubsection} Here we make contact with some of the results of \cite{Liao:2015lta}. We start with a connection on the tangent bundle $T\cM$. This induces a connection on $T^\star \cM$. We can be explicit in local coordinates:\footnote{Here $\nabla_a\equiv\nabla_{\pd_a}$, $\nabla_X$ for vector field $X$ is $\bbR$-linear, degree-preserving, and satisfies the obvious Leinbiz-esque rule, and we demand the duality relation involving the vector/1-form pairing: $Y(X^a\alpha_a)=(\nabla_Y X)^a\alpha_a+(-1)^{XY}X^a (\nabla_Y\alpha)_a\,.$}
\be
\nabla_c \pd_a\equiv A_{ca}^{b} \pd_b\implies \nabla_a \varepsilon^c=-\varepsilon^b (-1)^{ab} A_{ab}^{c} \,.
\ee
The connection is torsion-free\footnote{The torsion is defined as the following expression, $T(X,Y)\equiv \nabla_X Y-(-1)^{XY}\nabla_Y X -[X,Y]$ for $[\bullet,\bullet]$ the graded commutator. This is manifestly graded-antisymmetric and $\cin(\cM)$-linear in the sense $T(X,fY)=(-1)^{Xf}fT(X,Y)$.} when
\be
\label{torsionfree}A_{ab}^c=(-1)^{ab}A_{ba}^c\,,
\ee
which we will assume henceforth.

We use the connection to define a vector field $d_\nabla$ on $T[1]\cM\oplus \tfo\cM$ via the Leibniz rule, such that it acts as the de Rham differential on $\cin(T[1]\cM)$, and as $\nabla$ along $\cin(\tfo\cM)$. Explicitly:
\be
d_\nabla (z^a)\equiv\dr z^a\,,\qquad d_\nabla (\dr z^a)=0\,,\qquad d_\nabla(\varepsilon^a)=-\dr z^b\varepsilon^c A_{cb}^{a} \,,
\ee
where $A_{cb}^a$ depends on $\{z^a\}$.

We can now construct a proper formal exponential map by completing $d_\nabla$ into a vector field $D$ on $T[1]\cM\oplus \tfo\cM$ that satisfies the conditions of a Grothendieck connection. In terms of the expansion of $D$ in resolution degree, we identify $d_\nabla\equiv D_0$ and we set
\be
D_{-1}\equiv-\dr z^a \frac{\pd}{\pd\varepsilon^a}
\ee
as appropriate for a proper formal exponential map. Then the first two $D^2=0$ conditions (see \eqref{expandedflatness}) are satisfied; in particular
\be
[D_0,D_{-1}]=0\iff 0=[D_0,D_{-1}]\varepsilon^a=(-1)^b \dr z^b \dr z^c A^a_{cb}
\ee
which indeed vanishes because the connection $\nabla$ is torsion-free \eqref{torsionfree}. The third $D^2=0$ condition \eqref{expandedflatness} (which is $(D_{0})^2=-[D_{-1},D_{+1}]$) is a consequence of the ``algebraic Bianchi identity'' $R^a{}_{[bcd]}=0$ of the curvature tensor (where $[\cdots]$ is graded antisymmetrisation). This can be understood as follows: we have $[D_{-1},(D_0)^2]=0$ due to the algebraic Bianchi identity. Using the `counting' vector field $\epsilon_\mathrm{HT}$ of \eqref{epsilonHT}, we calculate
\be
[\epsilon_\mathrm{HT}, (D_0)^2]=2 (D_0)^2=[\zeta D_{-1}+D_{-1}\zeta, (D_0)^2]=[D_{-1},[\zeta,(D_0)^2]]\,,
\ee
whence we find that a suitable $D_{+1}$ is $-\frac{1}{2}[\zeta,(D_0)^2]$. At this point, we find that the assumptions of the homological perturbation result of the textbook \cite{henneaux1992quantization} of Henneaux and Teitelboim are satisfied\footnote{This is the same argument that is used to establish the existence of solutions to the BV-BRST classical master equation in field theory \cite{Fisch:1989rp,Fisch:1989rp}.}, so we can complete $D_{-1},D_{0},D_{+1}$ to a nilpotent
\be
D\equiv D_{-1}+D_{0}+D_{+1}+\dots
\ee
where the elided terms have $\operatorname{resdeg}\geq 2$.  We recall and apply the theorem in appendix \ref{app:HT}. We have thus constructed a proper formal exponential map via its associated Grothendieck connection.

We can also use the argument in this subsection to deduce that given \emph{any} proper formal exponential map, we have an associated connection on $T\cM$, that is torsion-free as a consequence of $D^2=0$. Given the expansion \eqref{Dexpansion} $D=D_{-1}+D_0+D_{+1}+\dots$ of the Grothendieck connection of a proper formal exponential map in terms of resolution degree, we recognise $D_0$ as a \emph{Quillen superconnection} associated to a connection on $T^\star \cM$ \cite{quillen1985superconnections}: $D_0$ indeed satisfies the Leibniz rule
\be
\begin{split}
D_0(\alpha f)= (\dr \alpha) f + (-1)^{\deg \alpha} \alpha D_0 f\,,\\ \alpha \in \cin(T[1]\cM)\,,\quad f\in \cin(T[1]\cM\oplus \tfo \cM)\,,\quad \resdeg f=1\,.
\end{split}
\ee
characterising such, if we identify differential forms with functions on $T[1]\cM$ as usual, and differential forms with values in $T^\star \cM$ with functions of resolution degree 1 on $T[1]\cM\oplus \tfo \cM$. This connection on $T^\star\cM$ then induces a unique connection on $T\cM$ as usual. We can summarise this as
\begin{theorem}
Every proper formal exponential map on a graded manifold $\cM$ gives rise to a unique torsion-free connection on the tangent bundle $T\cM$, and vice versa.
\end{theorem}

This has an interesting implication in the context of the construction of formal exponential maps from \emph{generalised exponential maps} $\exp:T\cM\to \cM$ (see subsection \ref{sec:fexpfromgexp} for the definition): for every generalised exponential map $\exp$ there exists a unique torsion-free connection on $T\cM$ whose (proper) formal exponential map $\fexp$ reproduces the vertical jets of $\exp$. In other words we can always find a connection that reproduces the Taylor coefficients in an expansion along the fibres of $T\cM$ of such $\exp$ evaluated on $\cM\into T\cM$.

\section{Linearising a QP-manifold at a point}
\subsection{Honest points}
\label{sec:QP}


For this section we assume the given graded manifold $\cM$ is also equipped with 
\begin{itemize}
    \item a degree 1 vector field $Q$ that squares to zero
    \item a degree $P$ symplectic form $\omega$
\end{itemize}
such that $(\cM,\omega,Q)$ has the structure of a \emph{QP-manifold} (also known as a differential graded symplectic manifold). It is of interest, for example in the context of the BV formalism of gauge theory, or in AKSZ topological field theories \cite{Alexandrov:1995kv}, to pass from a QP structure on $\cM$ to a QP structure on each tangent space, because the latter is linear; if the QP structure depends smoothly on the chosen point $x$, then one can in principle keep track of said dependence as the chosen point varies. (See \cite{Cattaneo:2001vu,Cattaneo:2018igh,Moshayedi:2020bkb} for applications where this is important.) We show how to construct a QP-structure on $T_x[\varepsilon]\cM$ using any formal exponential map $\fexp$ in the sense of Definition \ref{fexpdef}.

Let $x\in M$ be a point of the \emph{body} $M$ of $\cM$ (as opposed to a generalised point; we discuss those later). We have the inclusion $i_x:T_x[\varepsilon]\cM\into T[\varepsilon]\cM$ of the formal tangent space at $x$ into the formal tangent bundle (since $\tfo\cM$ is a vector bundle, this can be constructed as the pullback bundle $x^\star \tfo\cM$ along the map $x:\{\cdot\}\to \cM$). Given also a formal exponential map we can construct a map
\be
\fexp_x\equiv \fexp\circ i_x: T_x[\varepsilon]\cM\to \cM
\ee
and pullback forms to the formal tangent bundle at $x$. This produces a form $\fexp_x^\star\omega\equiv \omega_{x,\fexp}$.

Vector fields are trickier because they do not pullback. Since, however, we can invert $\fexp\circ i_x$ \emph{locally}, we can produce a vector field on each tangent space $T_x[\varepsilon]\cM$. Firstly we consider the case where 
\be
\fexp^\star z^a=z^a+\varepsilon^a
\ee
in local coordinates $\{z^a\}$ valid around the point $x$. If we denote $x^a\equiv i_x^\star z^a$ the vector of real coordinates of the point $x$ in the coordinate chart of $M$ arising from the chart of $\cM$, we can write $\fexp_x$ down explicitly along with its locally-defined inverse $(\fexp_x)^{-1}: \cM\to T_x[\varepsilon]\cM$:
\be
\fexp_x^\star z^a = x^a +\varepsilon^a\,,\qquad (\fexp_x^{-1})^\star \varepsilon^a= z^a-x^a\,.
\ee
Therefore if $Q$ is a derivation of $\cin(\cM)$, it gives rise to a derivation
\be
Q_{x,\fexp}\equiv\fexp_x^\star Q (\fexp_x^{-1})^\star
\ee
of $\cin(T_x[\varepsilon]\cM)$. For this it does not matter that $\fexp_x^{-1}$ is locally defined on $\cM$; in fact $Q$ itself only needs to be defined in some neighbourhood of $x$. The formulas for forms and vectors on the formal tangent space are compatible in the sense that contractions and Lie derivatives obviously go to contractions and Lie derivatives.

Therefore if $Q$ and $\omega$ define the Q and P-structures of a QP manifold, then their lifts to the formal tangent space at any point will again be compatible (where $L_Q$ is the Lie derivative):
\be
L_Q \omega=0\implies L_{Q_{x,\fexp}}\omega_{x,\fexp}=0\,.
\ee Moreover, $\omega_{x,\fexp}$ will define a nondegenerate symplectic form if $\omega$ does: if the coordinates $\{z^a\}$ on $\cM$ around $x$ are chosen to be Darboux, so $\omega=\frac{1}{2} \dr z^a \omega_{ab}\dr z^b$ with $\omega_{ab}$ a nondegenerate real matrix, then
\be
\omega_{x,\fexp}=\fexp_x^\star \omega=\tfrac{1}{2} \dr\varepsilon^a \omega_{ab}\dr \varepsilon^b\,.
\ee

Finally, to deal with the generic case $\fexp^\star z^a = z^a + \varepsilon^b e_b^a(z)+\dots$ we follow the strategy of the proof of Proposition \ref{flatnessprop} to replace $\fexp$ with one which takes the canonical form above. One can also show that the construction does not depend on the coordinate choices made using an argument similar to that of subsection \ref{sec:diffeo}. Since the formal tangent space is a \emph{formal pointed (graded) manifold} in the sense of \cite{Kontsevich:1997vb}, and Q-structures on such are equivalent to $L_\infty$-algebras, we find
\begin{theorem}
On a QP-manifold $(\cM,\omega,Q)$ equipped with a formal exponential map, given any point $x$ of the body $M\into \cM$ there exists an $L_\infty$-algebra structure on $T_x[\varepsilon]\cM$ equipped with an invariant non-degenerate inner product.
\end{theorem}
\noindent Since the $L_\infty$ structure is obviously smooth in $x$, we could say the formal tangent bundle is a bundle of $L_\infty$-algebras in this case.

\subsection{Remarks on generalised points}
\label{sec:generalisedpoints}
It is well-known that graded manifolds (and supermanifolds) ``do not have enough points'': if $\cM$ is a graded manifold with body $M$, we cannot recover it as the space of maps $\cin(\mathrm{point},\cM)$ from the point to $\cM$ unless $\cM=M$. (By $\cin(\cN,\cM)$ we denote (degree-preserving) morphisms of graded  rings $\cin(\cM)\to\cin(\cN)$; these are interpreted as pullbacks $\varphi^\star$ by maps $\varphi: \cN\to\cM$. We also have $\cin(\mathrm{point})=\bbR$.) We \emph{can} instead recover $\cM$ by considering \emph{generalised points}
\be
\label{eq:xgeneralisedpoint}
x: \cS\to \cM
\ee
for $\cS$ an arbitrary graded manifold (i.e.~morphisms $x^\star$ of graded rings $\cin(\cM)\to \cin(\cS)$). Then $\cM$ is identified with its \emph{functor of points} that maps an arbitrary $\cS$ to the space of maps $\cin(\cS,\cM)$. For more details we refer to the pedagogical discussion in  \cite[Section 5]{stolz2012equivariant} (in the context of supermanifolds).

We consider linearising about a generalised point. The role of the tangent space at a point ought to be played by the vector space of sections $\Gamma[x^\star \tfo\cM]$ of the pullback bundle\footnote{We can write the module of sections of the pullback bundle $\Gamma[x^\star \mathcal E]$ of any graded vector bundle $\mathcal E$ over $\mathcal M$ as the (completed/topological) tensor product $\cin(\mathcal S)\otimes_{\cin(\mathcal M)}\Gamma[\mathcal E]$, where $\cin(\mathcal S)$ is a $\cin(\mathcal M)$-module via the map $x^\star:\cin(\mathcal M)\to \cin(\mathcal S)$ (and $\Gamma[\mathcal E]$ is one by assumption). By the definition of the tensor product, if $s\in\cin(\mathcal S)$ and $e\in \Gamma[\mathcal E]$ then $$(x^\star f)\otimes e=(\mathrm{sign}) s\otimes(f\cdot e)\,.$$ The inclusion then corresponds to the obvious map of sections $\Gamma[\mathcal E]\to \cin(\mathcal S)\otimes_{\cin(\cM)}\Gamma[\mathcal E]$.} $x^\star \tfo\cM$ along the map $x$. However it is not entirely clear how to transport a QP-structure to this space of sections. One idea is to relate these pullback sections to generalised points: given the map $i_x: x^\star \tfo\cM\to \tfo\cM$ obtained by the pullback bundle construction,  we can again form $\fexp_x\equiv x^\star \tfo\cM\to \tfo\cM$ and obtain a map
\be
\Gamma[x^\star \tfo\cM] \to \cin(\cS,\cM)
\ee
that sends the section $\sfs \in \Gamma[x^\star \tfo\cM]$ to the generalised point $\fexp_x\circ\,\sfs:\cS\to \cM$. (This in particular sends the zero section to the original point $x$.) Given a non-degenerate integral on $\cS$, we can then at least transgress forms to $\cin(\cS,\cM)$ and then pullback to pullback sections. However this relies on the interpretation of sections of $x^\star \tfo\cM$ as maps $\cS\to x^\star \tfo\cM$ (as opposed to elements of $\cin(\cS)$-modules), which appears to be problematic for general $\cS$ due to the issue explained at the end of subsection \ref{subsection:tforesdeg}. For this reason we do not pursue linearisations of QP structures around generalised points further.

There is however an interesting implication that a choice of a formal exponential map has for a generalised point:
\begin{proposition}
The derivative of the map $\underline{\fexp_x}=(\fexp \circ i_x)\times (p_{x^\star \tfo}): x^\star \tfo \cM\to \cS\times \cM$ (where $p_{x^\star \tfo}:x^\star \tfo\cM\to \cS$ is the bundle projection) is invertible at the zero section $\sfs_0:\cS\to x^\star \tfo\cM$. The image of the zero section is the graph $g_x:\cS\to \cS\times \cM
$ of the generalised point $x:\cS\to\cM$.
\end{proposition}
\noindent Hence, for generalised points, one ought to think of a formal exponential map as a would-be diffeomorphism from a neighbourhood of the zero section in $x^\star \tfo\cM$ to a neighbourhood of the graph of $x$.
\begin{proof}
Since the image of $\sfs_0$ under $i_x:x^\star \tfo \cM\to \tfo\cM$ is the zero section of $\tfo\cM$, property \eqref{eq:fexp:def:propertyA} of formal exponential maps immediately implies the image of $\sfs_0$ is $g_x$. For the statement about derivatives, it suffices to calculate in local coordinates, where for $\cS$ we employ the local homogeneous coordinates $\{s^i\}$. We calculate the nonvanishing derivatives as
\be
\begin{split}
\sfs_0^\star \frac{\pd}{\pd \varepsilon^a} \underline{\fexp_x^\star} z^a= x^\star e_a^b(z)\,,\qquad \sfs_0^\star \frac{\pd}{\pd s^i} \underline{\fexp_x^\star} z^a= \frac{\pd x^\star z^a}{\pd s^i}\,,\\ \sfs_0^\star \frac{\pd}{\pd s^i} \underline{\fexp_x^\star} s^j=\delta^j_i\,,\qquad \sfs_0^\star \frac{\pd}{\pd \varepsilon^a} \underline{\fexp_x^\star} s^i=0\,.
\end{split}
\ee
The matrix of the derivative is this upper triangular with invertible diagonal blocks (due to \eqref{eq:fexp:def:propertyB}), hence it is invertible.
\end{proof}

\section*{Acknowledgements}
I would like to thank Jan Vysok\'y for patiently answering very basic questions about foundational issues of graded manifolds.

I am  supported of the FWO-Vlaanderen through the project G006119N, as well as by the Vrije Universiteit Brussel through the Strategic Research Program “High-Energy Physics”. I am also supported by an FWO Senior Postdoctoral Fellowship.

\appendix
\section{On the homological perturbation construction of proper formal exponential maps from connections}
\label{app:HT}


We will employ the following version of homological perturbation, which underlies the construction of solutions to the BV-BRST classical master equation:
\begin{theorem}[8.3 \cite{henneaux1992quantization}]
Let $\mathcal A$ be a graded-commutative algebra with odd left derivations $d$ and $\delta$. We assume $\mathcal A$ has two gradings $\operatorname{resdeg}$ and $\deg_{\mathrm{HT}}$ so that
\be
\operatorname{resdeg}\delta=-1,\quad \deg_{\mathrm{HT}}\delta=0,\quad \operatorname{resdeg}d=0,\quad \deg_{\mathrm{HT}}d=1\,.
\ee
(These gradings respectively correspond to antifield number and pureghost number in the physics context.) We also assume $\delta$ is a differential ($\delta^2=0$) and that $\delta x=0$ if $\operatorname{resdeg} x=0$; for $d$ we assume (in terms of the graded commutator)
\be
\label{eq:app:d:assumptions}
[d,\delta]=0\,,\qquad d^2=-[\delta,s_1]
\ee
for some odd left derivation $s_1$ of resolution degree 1.

If furthermore $\delta$ is acyclic ($H_{k\neq 0}(\delta)=0$)
then
\begin{itemize}
   \item[(a)] There exists a differential $s=\delta+d+s_1+\cdots$\footnote{The theorem requires assumptions to render infinite sums indexed over resolution degree well-defined; in the physics context this is usually ensured by the fact antifields are anticommuting. For the case of this paper, the sums are well-defined in the sense of formal power series.
   } of degree 1 in $  \deg_{\mathrm{HT}}-\operatorname{resdeg}$;
   \item[(b)] Any such $s$ has cohomology
   \be
   H^\bullet(s)=H^\bullet\big(d| H_0(\delta)\big)\,,
   \ee
   where $H^k\big(d| H_0(\delta)\big)$ is defined via the equivalence class $\sim$ below:
   \be
   \label{cohomologywithvalues:Def}
   \begin{split}
      x\in \mathcal A\,,\qquad \operatorname{resdeg} x=0\,,\qquad \deg_{\mathrm{HT}}x=\operatorname{formdeg}x=k\,,\\ \delta x=0\,, \qquad dx=\delta y\,,\qquad
      x\sim x + d z + \delta z'\,.
   \end{split}
   \ee
\end{itemize}
\end{theorem}
\noindent To satisfy the assumptions of the theorem for the purposes of the construction of subsection \ref{fexpconnsubsection}, we identify
\be
\deg_{\mathrm{HT}}=\operatorname{resdeg}+\operatorname{formdeg}\,.
\ee
Then we see that $\delta^2=0$ and the identities \eqref{eq:app:d:assumptions} are a consequence of the first three $D^2=0$ identities of \eqref{expandedflatness} which were established in subsection \ref{fexpconnsubsection} already. We then obtain the new differential as $s\equiv D$. This is of degree $+1$ in the total degree $\operatorname{formdeg}$ (this would be ghost number in the physics context).

\let\oldbibliography\thebibliography
\renewcommand{\thebibliography}[1]{%
    \oldbibliography{#1}%
    \setlength{\itemsep}{-1pt}%
}
\begin{multicols}{2}
{\setstretch{0}
    \small
    \bibliography{Alex_notes.bib}}
\end{multicols}
\end{document}